\documentclass[letterpaper,10pt,conference]{ieeeconf}

\IEEEoverridecommandlockouts
\newtheorem{theorem}{Theorem}

\newtheorem{problem}{Problem}

\usepackage{csquotes}
\usepackage{amssymb}
\usepackage{graphicx}
\usepackage[cmex10]{amsmath}
\usepackage{array}
\usepackage{mdwtab}
\usepackage{fixltx2e}
\usepackage{graphics}
\usepackage{epsfig}
\usepackage{times}
\usepackage{amsmath}
\usepackage{amssymb}
\usepackage{url}
\usepackage{lscape}
\usepackage{color}
\usepackage{epsfig}
\usepackage{cite}
\usepackage{balance}
\usepackage{dblfloatfix}
\usepackage{hyperref}

\title{\LARGE \bf Active Attack Detection and Control in Constrained Cyber-Physical Systems Under Prevented Actuation Attack}

\author{Mehdi Hosseinzadeh, \IEEEmembership{Member,~IEEE}, and Bruno Sinopoli, \IEEEmembership{Fellow,~IEEE}
\thanks{This research has been supported by National Science Foundation under award numbers ECCS-1932530.}
\thanks{M. Hosseinzadeh and B. Sinopoli are with the Department of Electrical and Systems Engineering, Washington University in St. Louis, St. Louis, Missouri, USA (email: mehdi.hosseinzadeh@ieee.org; bsinopoli@wustl.edu).}
}

\begin{document}

\maketitle
\thispagestyle{empty}
\pagestyle{empty}

\begin{abstract}
This paper proposes an active attack detection scheme for constrained cyber-physical systems. Despite passive approaches where the detection is based on the analysis of the input-output data, active approaches interact with the system by designing the control input so to improve detection. This paper focuses on the prevented actuation attack, where the attacker prevents the exchange of information between the controller and actuators. The proposed scheme consists of two units: 1) detection, and 2) control. The detection unit includes a set of parallel detectors, which are designed based on the multiple-model adaptive estimation approach to detect the attack and to identify the attacked actuator(s). For what regards the control unit, a constrained optimization approach is developed to determine the control input such that the control and detection aims are achieved. In the formulation of the detection and control objective functions, a probabilistic approach is used to reap the benefits of the \textit{a priori} information availability. The effectiveness of the proposed scheme is demonstrated through a simulation study on an irrigation channel.
\end{abstract}

\section{Introduction}\label{sec:introduction}
Cyber-Physical Systems (CPSs) often employ distributed networks of embedded sensors and actuators that interact with the physical environment. The availability of cheap communication technologies (e.g., internet) has certainly improved scalability and functionality features in several applications. However, they have made CPSs susceptible to cyber security threats. This makes the cyber security to be of primary importance in safe operation of CPSs.


By assuming that sensors-to-controller and controller-to-actuators communication channels are the only ones in CPSs executed via internet and malicious agents can alter data flows in these channels, in general two classes of cyber attacks can be considered: (i) False Data Injection (FDI), and (ii) Denials of Service (DoS). A FDI (a.k.a. deception attack) affects the data integrity of packets by modifying their payloads \cite{Mo2014SignalProcessing,Mo2010Conf,Bai2017}. A DoS is the one that the attacker needs only to disrupt the system by preventing communication between the components. In this paper, we focus on a specific type of DoS attack, so-called Prevented Actuation Attack (PA2) \cite{Loukas2015,Carvalho2016}, where the attacker prevents the exchange of information between the controller and the actuators. An attacker can launch such attacks on the physical layer or cyber layer. Examples of real-world PA2 are: sleep deprivation torture attack \cite{Stajano1999} (a.k.a. battery exhaustion attack) that exhausts the battery of a surveillance robot or a medical implant until it can no longer function; door lock attack \cite{Ha2017} that suppresses the operation of a smart door by injecting `close' command every time an `open' command is received; and fatigue bearing attack \cite{Wu2019} that restrains the operation of the lubricant system in wind turbines to damage gearboxess.

Regardless of the type of attack, attack detection approaches presented in the literature can be classified as: (i) passive approaches, and (ii) active approaches. Note that we use the same terminology of the fault literature \cite{HosseinzadehPV} to classify attack detection approaches, as faults and attacks usually manifest themselves similarly in control systems despite their natural differences. In passive approaches, the input-output data of the system are measured (remotely or on-site), analyzed for any possible stealthy behavior, and then a decision about an attack is made. The passive approaches are widely studied and commonly used in many today's applications, e.g., \cite{Mo2014TCST,Zhang2016,Li2017,Li2018,Zhang2019}. However, they might not be able to recognize an attack when the input-output data are not informative enough. Also, they do not address stability/safety of the system during \textit{detection horizon}, a time interval from the instant an attack occurs to the instant when it is detected. 

The active approaches interact with the system during the detection horizon by means of a suitably designed input signal that is injected into the system to increase the quality of detection, shorten the detection horizon, and enforce stability/safety during the detection horizon. Contrary to the passive approaches, the active approaches are historically younger and still under development. To the best of the authors' knowledge, the only existing active attack detection approach in the literature is the physical authentication (a.k.a. digital watermarking) \cite{Mo2015,Irita2017,Sanchez2019}. The core idea of this method is to inject a known noisy input to the system and observe its effect on the output of the system. Thus, if an attacker is unaware of this physical watermark, the system cannot be adequately emulated, as the attacker is unable to consistently generate the component of the output associated with this known noisy input. The physical authentication, which is mainly used in detection of replay attack (a.k.a. playback attack) \cite{Mo2009}, can be effective if the noise injected at the system input is large enough to achieve good detection performance, which may degrade the control performance. Moreover, this method injects the noisy input irrespective of the probability of attack occurrence, which leads to unneeded loss in control performance. Furthermore, in the case of constrained systems \cite{Hosseinzadeh2019,HosseinzadehCDC2020,HosseinzadehTAC}, as shown in \cite{Hosseinzadeh2019Allerton}, the extra uncertainty injected to the system due to the noisy input should be taken into account in the design procedure, which leads to tighter constraints, and consequently more conservative behavior.

This paper answers the following question: \textit{How to determine the control input sequence for a constrained CPS such as to improve the detection performance without degrading the control performance?} Inspired by \cite{Puncochar2015}, this paper answers this question in the case of PA2. The proposed structure consists of two units: (i) detection unit, and (ii) control unit. The detection unit uses \textit{a priori} information and input-output data of the system over the detection horizon with a certain length to generate a decision variable which represents the situation of the system. More precisely, the detection unit recognizes the existence/inexistence of PA2 and distinguishes attacked actuators. The control unit generates the control input which is optimal according to a cost function and guarantees constraint satisfaction at all times. Both control and detection aims are defined in the form of stochastic objective functions, i.e., they are uncertain due to noises and initial condition. The open-loop information processing strategy \cite{Simandl2009} is then used to express the stochastic objective functions as deterministic functions. Finally, in order to evaluate the quality of the control input sequence in terms of detection and control aims, a compromise between the two aims is defined in the form of a multi-objective optimization problem whose solution can be computed by means of available optimization tools.


\section{Problem Statement}\label{sec:ProblemStatement}
Consider the following discrete-time LTI system:
\begin{align}
x_{k+1}=&Ax_k+Bu_k+w_k, \label{eq:modelfree1}  \\
y_k=&Cx_k+v_k, \label{eq:modelfree2}
\end{align}
where $x_k\in\mathbb{R}^n$ is the state vector at time $k$, $u_k\in\mathbb{R}^p$ is the control input at time $k$, $y_k\in\mathbb{R}^m$ is the vector of measurements from the sensors at time $k$, and the process noise $w_k\in\mathbb{R}^n$ and the measurement noise $v_k\in\mathbb{R}^m$ are mutually independent white Gaussian noises with zero mean and covariance matrices $H_w\in\mathbb{R}^{n\times n}$ and $H_v\in\mathbb{R}^{m\times m}$, respectively. We assume that the initial state $x_0$ is independent of $w_k$ and $v_k$, and has a Gaussian distribution with the known mean $\bar{x}_0$ and covariance matrix $H_{x,(0,0)}$.

In mathematical terms, the PA2 on the $i$-th actuator is equivalent to zeroing the $i$-th column in the matrix $B$. Thus, the dynamics of attack-free and under attack systems can be expressed using a single difference equation in the following form:
\begin{align}
x_{k+1|\mu}=&Ax_{k|\mu}+B_\mu u_k+w_k, \label{eq:modelfree3}\\
y_{k|\mu}=&Cx_{k|\mu}+v_k, \label{eq:modelfree4}
\end{align}
where $\mu\in\{\mu_1,\cdots,\mu_{2^p}\}$ is the index of the mode of the system, each $\mu_i$ having a known distribution $P(\mu_i)$, $B_\mu$ is the corresponding input matrix, $x_{k|\mu}$ is the state of the system operating in mode $\mu$ with $x_{0|\mu}=x_0$, and $y_{k|\mu}$ is the output of the system operating in mode $\mu$. 



Suppose that the system is subject to the following expectational linear constraints:
\begin{equation}\label{eq:constraints}
\text{E}\left[G_{x}x_{k|\mu_i}+G_{u}u_k\right]\leq g,~\forall k\geq0,~i\in\{1,\cdots,2^p\}
\end{equation}
where $\text{E}[\cdot]$ is the expectation function, and $G_{x}\in\mathbb{R}^{n_c\times n}$, $G_{u}\in\mathbb{R}^{n_c\times p}$, and $g\in\mathbb{R}^{n_c}$, with $n_c$ as the number of constraints.

\begin{problem}\label{prob:mainprob}
Consider system \eqref{eq:modelfree3}-\eqref{eq:modelfree4} which is subject to constraints \eqref{eq:constraints}. Suppose $N>0$ as the detection horizon, chosen by the designer. Find a control sequence $u_k,~k=0,\cdots,N-1$ such that\footnote{Without loss of generality and for the sake of simplicity we assume the detection horizon starts from 0.} at time $N$ the mode of the system is identified with high probability of correctness, while optimal control performance and constraint satisfaction are guaranteed during the detection horizon. 
\end{problem}


Before starting with the solution of Problem \ref{prob:mainprob}, let us compute the conditional probability density functions of the state and output. According to \eqref{eq:modelfree3}, for the control sequence $u_0,\cdots,u_{N-1}$, the mean value of the state at time $k$ is 
\begin{align}
\bar{x}_{k|\mu}=A^k\bar{x}_{0|\mu}+A^{k-1}B_\mu u_0+\cdots+B_\mu u_{k-1},\label{eq:xbar}
\end{align}
and the covariance matrix of the state at times $k$ and $l$ ($k\geq l$) can be computed as
\begin{align}
H_{x,(k,l)|\mu}:=&\text{E}\left\{(x_{k|\mu}-\bar{x}_{k|\mu})(x_{l|\mu}-\bar{x}_{l|\mu})^T\right\}\nonumber\\
=&A^kH_{x,(0,0)}\left(A^l\right)^T+A^{k-1}H_w\left(A^{l-1}\right)^T\nonumber\\
&+A^{k-2}H_w\left(A^{l-2}\right)^T+\cdots+A^{k-l}H_w\label{eq:Hxbar},
\end{align}
where $H_{x,(k,l)|\mu}=\left(H_{x,(l,k)|\mu}\right)^T\in\mathbb{R}^{n\times n}$. Therefore, the conditional probability density function of the state for the interval $[0,N]$ can be expressed as:
\begin{align}
P\left(x_{0:N}\big|\mu,u_{0:N-1}\right)\sim \mathcal{N}\left(\bar{x}_{0:N|\mu},H_{x|\mu}\right),
\end{align}
where $x_{0:N}:=[(x_0)^T,\cdots,(x_N)^T]^T\in\mathbb{R}^{n(N+1)}$, $u_{0:N-1}:=[(u_0)^T,\cdots,(u_{N-1})^T]^T\in\mathbb{R}^{pN}$, $\bar{x}_{0:N|\mu}:=[(\bar{x}_{0|\mu})^T,\cdots,(\bar{x}_{N|\mu})^T]^T\in\mathbb{R}^{n(N+1)}$, and $H_{x,(i-1,j-1)|\mu}$ is the element $(i.j)$ of $H_{x|\mu}\in\mathbb{R}^{n(N+1)\times n(N+1)}$.

Similarly, the conditional probability density function of the output for the the interval $[0,N]$ can be expressed as
\begin{align}
P\left(y_{0:N}\big|\mu,u_{0:N-1}\right)\sim \mathcal{N}\left(\bar{y}_{0:N|\mu},H_{y|\mu}\right),
\end{align}
where $y_{0:N}:=[(y_0)^T,\cdots,(y_N)^T]^T\in\mathbb{R}^{m(N+1)}$, and $\bar{y}_{0:N|\mu}:=[(\bar{y}_{0|\mu})^T,\cdots,(\bar{y}_{N|\mu})^T]^T\in\mathbb{R}^{m(N+1)}$ with $\bar{y}_{k|\mu}=C\bar{x}_{k|\mu},~k=0,\cdots,N$ as the mean value of the output at time $k$. Also, $H_{y|\mu}\in\mathbb{R}^{m(N+1)\times m(N+1)}$ is the covariance matrix of the output, where the $(i,j)$ element is
\begin{align}
H_{y,(i,j)|\mu}=\left\{
\begin{array}{rl}
CH_{x,(i-1,j-1)|\mu}C^T, & i>j\\
CH_{x,(i-1,j-1)|\mu}C^T+H_v & i=j
\end{array}
\right..\label{eq:Hybar}
\end{align}

\section{Detection Unit}\label{sec:detectionunit}
The system \eqref{eq:modelfree3}-\eqref{eq:modelfree4} can be seen as a $2^p$-model system, where each model corresponds to one mode. Thus, in order to detect existence/inexistence of PA2, and to identify which actuator(s) is under attack, it is only needed to identify the true $\mu$ at the end of detection horizon. The fact that one of the $2^p$ models is the true one can be modeled by a hypothesis random variable that must belong to a discrete set of hypothesis $\{\mu_1,\cdots,\mu_{2^p}\}$, where the event $\mu_i$ means that the $i$-th model is the one that is generating the data.

One Bayesian approach to hypothesis testing is to base decisions on the posterior probabilities, i.e., the probability of the mode $\mu_i$ conditioned by the input-output data. In mathematical terms, the posterior probabilities at time $k$ are denoted as $P\left(\mu_i\big|y_{0:k},u_{0:k-1}\right)$, where at time $k=0$ the posterior probabilities are equal to the prior probabilities, i.e., $P\left(\mu_i\big|u_0^T\right)=P(\mu_i)$.

The conditioned posterior probabilities can be computed using the Multiple-Model Adaptive Estimator (MMAE) structure \cite{Fekri2004,Hassani2009,Sadati2018}. The MMAE (a.k.a. partitioned algorithm) involves the parallel operation of $2^p$ Kalman filters (each matched to one of the postulated models), where the residuals of the Kalman filters are used to compute the conditional posterior probabilities. The rationale is that the highest posterior probability corresponds to the true model of the system. It is shown that the correct model  can be identified ``almost surely" \cite{Fekri2006,Rotondo2017}.

It is easy to show that when the attack happens sometime within a detection horizon, it might remain undetected until the end of the following detection horizon. It can be also shown that in the case of a smart attack (i.e., the attack happens sometime within a detection horizon and lasts for a wisely selected period of time), a single detector might not be adequate to detect the attack. Therefore, we propose to deploy $N$ parallel detectors, where the $d$-th detector identifies the mode $\hat{\mu}^d$ via
\begin{equation}\label{eq:detector}
\hat{\mu}^d=\arg\max\limits_{i\in\{1,2,\cdots,2^p\}}P\left(\mu_i\big|y_{0:N},u_{0:N-1}\right).
\end{equation}



We assume that every detector identifies the mode of the system only in $N$ time steps. Note that $N$ defines the the trade-off between the detection quality and detection performance. Large values of $N$ decreases the probability of making an incorrect decision during the transient of the posterior probabilities. However, when the attack duration is too small compared to the length of the detection horizon, the attack might remain undetected.

\section{Control Unit}\label{sec:ControlUnit}

\subsection{Control Objective Function}
The control aim is to track the desired reference $r_k\in\mathbb{R}^m$ while penalizing the control effort. The control objective function can be formulated as
\begin{align}
J_c(u_{0:N-1})=\text{E}\Bigg[&\sum\limits_{k=0}^{N}\left\Vert y_k-r_k\right\Vert^2_Q+\sum\limits_{k=0}^{N-1}\left\Vert u_k\right\Vert^2_R\Bigg]\label{eq:ControlOF1}
\end{align}
where $Q=Q^T\in\mathbb{R}^{m\times m}$ is a positive semi-definite matrix and $R=R^T\in\mathbb{R}^{p\times p}$ is a positive definite matrix. 

\begin{figure*}[b]
\hrule
\setcounter{equation}{17}
\begin{align}
F_1=&\left[\begin{matrix}B_{\mu_i}^T(A^{k-1})^TC^TQCA^{k-1}B_{\mu_i}&\cdots&B_{\mu_i}^T(A^{k-1})^TC^TQCB_{\mu_i}\\\vdots&\ddots&\vdots\\(B_{\mu_i})^T C^TQCA^{k-1}B_{\mu_i}&\cdots&(B_{\mu_i})^TC^TQCB_{\mu_i}\end{matrix}\right],\label{eq:F1}\\
F_2=&2(\bar{x}_{0|\mu_i})^T(A^k)^TQ\left[\begin{matrix}A^{k-1}B_{\mu_i}&A^{k-2}B_{\mu_i}&\cdots&B_{\mu_i}\end{matrix}\right]-2r_k^TQC\left[\begin{matrix}A^{k-1}B_{\mu_i}&A^{k-2}B_{\mu_i}&\cdots&B_{\mu_i}\end{matrix}\right],\label{eq:F2}\\
F_3=&\sum\limits_{k=0}^N\sum\limits_{i=0}^{2^p}P(\mu_i)(\bar{x}_{0|\mu_i})^T(A^k)^TC^TQCA^k\bar{x}_{0|\mu_i}+\text{Tr}\left(Q\sum\limits_{k=0}^{N}\sum\limits_{i=1}^{2^p}P(\mu_i)H_{y,(k,k)|\mu_i}\right)+\sum\limits_{k=0}^{N}r_k^TQr_k\nonumber\\
&+\sum\limits_{k=0}^N\sum\limits_{i=0}^{2^p}P(\mu_i)r^TQCA^k\bar{x}_{0|\mu_i}.\label{eq:F3}
\end{align}
\end{figure*}

The objective function \eqref{eq:ControlOF1} is a stochastic function, where the uncertainties are due to noises and the initial condition. In this paper, instead of using deterministic approaches (i.e., assuming uncertainties as upper-bounded signals), we will focus on probabilistic approaches, where available \textit{a priori} information can be used in obtaining the optimal control sequence. In particular, we will use the open loop approach. This approach is based on the information available at the beginning of each detection horizon (i.e., $\bar{x}_0$ and $H_{x,(0,0)}$), while the measurements during the detection horizon are not used. 

\begin{theorem}\label{theorem1}
Consider system \eqref{eq:modelfree3}-\eqref{eq:modelfree4}, and 
control objective function \eqref{eq:ControlOF1}. Suppose that the open loop approach is used to determine the control sequence, i.e., the entire control sequence $u_{0:N-1}$ is determined at the beginning of the prediction horizon. Then, control objective function \eqref{eq:ControlOF1} can be expressed as an explicit function of the control sequence.
\end{theorem}

\begin{proof}
When the open loop approach is used, the objective function \eqref{eq:ControlOF1} can be expressed as \setcounter{equation}{12}
\begin{align}\label{eq:ControlOF2}
J_c(\cdot)=&\text{Tr}\Bigg(Q\sum\limits_{k=0}^{N}\text{E}\Big[y_ky_k^T\Big|u_{0:N-1}\Big]\Bigg)+\sum\limits_{k=0}^{N}r_k^TQr_k\nonumber\\
&-2\sum\limits_{k=0}^{N}r_k^TQ\text{E}\Big[y_k\Big|u_{0:N-1}\Big]+\sum\limits_{k=0}^{N-1}u_k^TRu_k,
\end{align}
where $\text{Tr}(\cdot)$ is the trace function. We know that\footnote{$\text{Cov}(Y,Y)=\text{E}\{YY^T\}-\text{E}\{Y\}\left(\text{E}\{Y\}\right)^T$ for the random vector $Y$.}
\begin{align}
\text{E}\Big[y_ky_k^T\big|u_{0:N-1}\Big]&=\sum\limits_{i=1}^{2^p}P(\mu_i)\left(\bar{y}_{k|\mu_i}\bar{y}_{k|\mu_i}^T+H_{y,(k,k)|\mu_i}\right),
\end{align}
where $\text{Cov}(\cdot)$ is the covariance function. Thus, the control objective function can be rewritten as
\begin{align}
J_c(\cdot)=&\text{Tr}\left(Q\sum\limits_{k=0}^{N}\sum\limits_{i=1}^{2^p}P(\mu_i)\left(\bar{y}_{k|\mu_i}\bar{y}_{k|\mu_i}^T+H_{y,(k,k)|\mu_i}\right)\right)\nonumber\\
&+\sum\limits_{k=0}^{N-1}u_k^TRu_k+\sum\limits_{k=0}^{N}r_k^TQr_k\nonumber\\
&-2\sum\limits_{k=0}^{N}\sum\limits_{i=1}^{2^p}P(\mu_i)r_k^TQ\bar{y}_{k|\mu_i},
\end{align}
which due to the fact that $\bar{y}_{k|\mu}=C\bar{x}_{k|\mu}$, it implies that:
\begin{align}\label{eq:Jcybar}
J_c(\cdot)=&\sum\limits_{k=0}^{N}\sum\limits_{i=1}^{2^p}P(\mu_i)\bar{x}_{k|\mu_i}^TC^TQC\bar{x}_{k|\mu_i}+\sum\limits_{k=0}^{N-1}u_k^TRu_k\nonumber\\
&-2\sum\limits_{k=0}^{N}\sum\limits_{i=1}^{2^p}P(\mu_i)r_k^TQC\bar{x}_{k|\mu_i}+\sum\limits_{k=0}^{N}r_k^TQr_k\nonumber\\
&+\text{Tr}\left(Q\sum\limits_{k=0}^{N}\sum\limits_{i=1}^{2^p}P(\mu_i)H_{y,(k,k)|\mu_i}\right).
\end{align}

Finally, according to \eqref{eq:xbar}, \eqref{eq:Jcybar} can be rewritten as
\begin{align}\label{eq:ControlOFFinal}
J_c(\cdot)=&\sum\limits_{k=0}^{N}\sum\limits_{i=1}^{2^p}P(\mu_i)u_{0:k-1}^TF_1u_{0:k-1}+\sum\limits_{k=0}^{N-1}u_k^TRu_k\nonumber\\
&+\sum\limits_{k=0}^{N}\sum\limits_{i=1}^{2^p}P(\mu_i)F_2u_{0:k-1}+F_3,
\end{align}
where $F_1$, $F_2$, and $F_3$ are given in \eqref{eq:F1}-\eqref{eq:F3}, respectively. This completes the proof.
\end{proof}

\subsection{Detection Objective Function}
Suppose that the detector given in \eqref{eq:detector} is used to identify the mode of the system. In order to determine the control sequence during the detection horizon such that the probability of an incorrect identification is minimized, we can use the following detection objective function\setcounter{equation}{20}
\begin{equation}\label{eq:detectionobjectivefunction}
J_d(u_{0:N})\triangleq \text{E}\big[\sigma(\hat{\mu}^1)\big],
\end{equation}
where $\sigma(\hat{\mu}^1)$ is zero when the identified mode is the actual mode of the system, and is non-zero (we set it to 1 for simplicity) otherwise. Note that since the control sequence $u_{0:N-1}$ is assumed to be applied at time $k=0$, only the first detection horizon is taken into account in the formulation of the detection objective function. In other words, the determined control signal is optimal only for Detector\#1.

\begin{theorem}\label{theorem2}
Consider system \eqref{eq:modelfree3}-\eqref{eq:modelfree4}, and 
detection objective function \eqref{eq:detectionobjectivefunction}. Suppose that the open loop approach is used to determine the control sequence. Then, the detection objective function can be upper bounded with an explicit function of the control sequence.
\end{theorem}

\begin{proof}
By using the open loop approach, the detection objective function \eqref{eq:detectionobjectivefunction} can be expressed as
\begin{align}
J_d(\cdot)=&\text{E}\big[\sigma(\hat{\mu}^1)\big|u_{0:N-1}\big]\nonumber\\
=&\int\limits_{\mathbb{R}^{m(N+1)}}\sum\limits_{i=1}^{2^p}\sigma(\hat{\mu}^1)P\left(\mu_i\big|y_{0:N},u_{0:N-1}\right)\cdot\nonumber\\
&~~~~~~~~~~~~~~~~~~P(y_{0:N}\big|u_{0:N-1})dy_{0:N},
\end{align}
which according to Bayes' theorem, it implies that
\begin{align}
J_d(\cdot)=&\int\limits_{\mathbb{R}^{m(N+1)}}\sum\limits_{i=1}^{2^p}\sigma(\hat{\mu}^1)P(y_{0:N}\big|\mu_i,u_{0:N-1})P\left(\mu_i\right)dy_{0:N},\label{eq:Jd1}
\end{align}
which is concluded due to the fact\footnote{$P(\mu_i|u_{0:N-1})=\frac{P(u_{0:N-1}|\mu_i)}{P(u_{0:N-1})}P(\mu_i)=P(\mu_i)$, since $u_{0:N-1}$ is deterministic, and consequently $P(u_{0:N-1}|\mu_i)=P(u_{0:N-1})=1$.} that the probability of the mode $\mu_i$ conditioned by only input data is equal to the probability of the mode $\mu_i$.

The right side of \eqref{eq:Jd1} cannot be computed analytically and its numerical evaluation is computationally expensive. Due to this reason, in the following we will find an upper bound for the detection objective function $J_d(u_0,\cdots,u_{N-1})$.

Since $0\leq\sigma(\hat{\mu}^1)\leq1$, it implies that:
\begin{align}
J_d(\cdot)\leq&\int\limits_{\mathbb{R}^{m(N+1)}}\sum\limits_{i=1}^{2^p} P(y_{0:N}\big|\mu_i,u_{0:N-1})P\left(\mu_i\right)dy_{0:N}.\label{eq:Jd2}
\end{align}

Following the same arguments presented in \cite{Blackmore2006}, the right side of \eqref{eq:Jd2} can be upper bounded as
\begin{align}
\int\limits_{\mathbb{R}^{m(N+1)}}\sum\limits_{\mu=1}^{2^p} P(y_{0:N}\big|\mu_i,u_{0:N-1})P\left(\mu_i\right)dy_{0:N}\leq \hat{J}_d(u_{0:N}),
\end{align}
where
\begin{align}
\hat{J}_d(u_{0:N})=\sum\limits_{i=1}^{2^p}\sum\limits_{j=i+1}^{2^p}\sqrt{P(\mu_i)P(\mu_j)}e^{-\phi_{ij}},\label{eq:Jdhat}
\end{align}
with 
\begin{align}
\phi_{ij}&=\frac{1}{4}\left(\bar{y}_{0:N|\mu_j}-\bar{y}_{0:N|\mu_i}\right)^T\left(H_{y|\mu_i}+H_{y|\mu_j}\right)^{-1}\cdot\nonumber\\
&\left(\bar{y}_{0:N|\mu_j}-\bar{y}_{0:N|\mu_i}\right)+\frac{1}{2}\ln\left(\frac{\text{det}\left(\frac{H_{y|\mu_i}+H_{y|\mu_j}}{2}\right)}{\sqrt{\text{det}(H_{y|\mu_i})\text{det}(H_{y|\mu_j})}}\right)\label{eq:phi}
\end{align}
where $\text{det}(\cdot)$ is the determinant function. It is noteworthy that according to \eqref{eq:xbar}-\eqref{eq:Hxbar}, \eqref{eq:Hybar}, and since $\bar{y}_{k|\mu_i}=C\bar{x}_{k|\mu_i}$, the upper bound $\hat{J}_d(\cdot)$ given in \eqref{eq:Jdhat} is an explicit function of the control sequence. This completes the proof.
\end{proof}

\subsection{Constraints}
By using the open loop approach, the expectational constraints given in \eqref{eq:constraints} take the following form:
\begin{align}
&\text{E}\left[G_{x}x_{k|\mu_i}+G_{u}u_k\big|u_{0:k-1}\right]\leq g\Rightarrow\nonumber\\
&G_{x}(A^k\bar{x}_{0|\mu_i}+A^{k-1}B_{\mu_i}u_0+\cdots+B_{\mu_i}u_{k-1})\nonumber\\
&+G_{u}u_k\leq g,\label{eq:constraintsFinal}
\end{align}
which is an explicit function of the control sequence.

\subsection{Proposed Solution}\label{sec:solution}

One possible way to pursue both control and detection aims is to let one of the objective functions to take arbitrary value up to a known upper limit value, and then to enforce this as a constraint and minimize the other objective function. Therefore, the following two optimization problems can be considered
\begin{equation}\label{eq:OptProb1}
u_{0:N-1}^\ast=\left\{
\begin{array}{ll}
& \arg\,\min\limits_{u_{0:N-1}}J_c(\cdot)\text{ given in \eqref{eq:ControlOFFinal}}\\
\text{s.t.} & \eqref{eq:constraintsFinal}\text{ is satisfied  }\forall i,~\forall k\geq0\\
& \hat{J}_d(\cdot)\text{ given in \eqref{eq:Jdhat}}\leq\bar{J}_d
\end{array}
\right.,
\end{equation}
or
\begin{equation}\label{eq:OptProb2}
u_{0:N-1}^\ast=\left\{
\begin{array}{ll}
& \arg\,\min\limits_{u_{0:N-1}}\hat{J}_d(\cdot)\text{ given in \eqref{eq:Jdhat}}\\
\text{s.t.} & \eqref{eq:constraintsFinal}\text{ is satisfied  }\forall i,~\forall k\geq0\\
& J_c(\cdot)\text{ given in \eqref{eq:ControlOFFinal}}\leq\bar{J}_c
\end{array}
\right.,
\end{equation}
where $\bar{J}_d$ and $\bar{J}_c$ are maximum acceptable levels of the detection and control objective functions, respectively. 


The objective function \eqref{eq:ControlOFFinal} and constraints given in \eqref{eq:constraintsFinal} are convex in $u_{0:N-1}$. The objective function \eqref{eq:Jdhat} is concave, as $\phi_{ij}$ as in \eqref{eq:phi} is a quadratic function of $u_{0:N-1}$ (with a positive definite matrix), and consequently convex in $u_{0:N-1}$. Thus, problems \eqref{eq:OptProb1}-\eqref{eq:OptProb2} are in general non-convex. We use \texttt{bmibnb} \cite{Lofberg2004} to numerically compute their solutions. 

\section{Simulation Study-- Irrigation Channel}\label{sec:simulation}
In this section we will use the developed method to control the level of water in pools 9 and 10 of the Haughton main channel, as shown in Fig. \ref{fig:WaterCanal}. The water levels in the channel are controlled by overshot gates located along the channel. The stretch of a channel between two gates is referred to as a reach or a pool. We assume that the communication between the controller and the gates is through internet. 


The water level in the $g$-th pool ($g\in\{9,10\}$) of the irrigation channel can be modeled as \cite{Zhang2005} \setcounter{equation}{30}
\begin{align}
\dot{y}_g(t)=\alpha_{g-1,in}h_{g-1}^{3/2}(t-\tau_{g-1})-\alpha_{g,out}h_{g}^{3/2}(t)+d_{g-1}(t),
\end{align}
where $y_g(t)$ is the water level in the pool, $h_g(t)$ is the head over the gate (the height of water above the gate), $\tau_g$ is the time delay which accounts for the time it takes for the water to travel from the upstream gate to the downstream gate in the $g$-th pool, $d_g(t)$ represents offtakes to farms and side channels, and $\alpha_{g,in}$ and $\alpha_{g,out}$ are constants which incorporates the effect of the discharge coefficients. The real value of the parameters is given in TABLE \ref{tab:Parameters}. For the sake of simplicity we assume there is no offtake, i.e., $d_g=0,~g=8,9$.

The sampling time is 10 [min], and the control signal is the head over the gate. We assume that initial water level in pool 9 and 10 is 6.60 [m] and 5.60 [m], respectively. Also, we assume that $w_k\sim\mathcal{N}\left(\textbf{0},0.3I_8\right)$ and  $v_k\sim\mathcal{N}\left(\textbf{0},0.3I_2\right)$, where $I_2$ is the $2\times2$ identity matrix and $\textbf{0}$ is the zero vector with appropriate size. The water level in pools should not exceed 15 [m]. The system is subject to actuator saturation \cite{Hosseinzadeh2017}, i.e., the control signals cannot be negative.

Since there are three actuators, eight different modes can be defined. We assume that \textit{a priori} probability of the $i$-th mode is  $P(\mu_i)=0.125,~\forall i$.

Suppose that $Q=I_2$, $R=I_3$, the detection horizon is 200 [min], and the level of detection and control objective functions must not exceed $1$ and $2000$, respectively.

\begin{figure}[t]
\centering
\includegraphics[width=8cm]{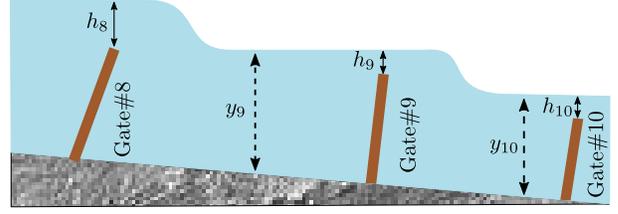}
\caption{Side view of the Haughton main channel; pools 9 and 10.}
\label{fig:WaterCanal}
\end{figure}

\begin{table}[!t]
\centering
\caption{Parameters of the Haughton Main Channel \cite{Weyer2001}.}\label{tab:Parameters}
\footnotesize
\begin{tabular}{c|c|c|c}
Parameter & $g=8$ & $g=9$ & $g=10$ \\
\hline
$\alpha_{g,in}$ [1/m$^2$]& 0.0208 & 0.0700 & 0.0142\\
$\alpha_{g,out}$ [1/m$^2$]& 0.0278 & 0.0614 & 0.0156\\
$\tau_{g}$ [min]& 6 & 3 & 16
\end{tabular}
\end{table}

We assume that for $k\in[0,80]$, $[200,300]$, $[360,480]$, $[580,700]$ the mode of the system is 1, for $k\in[80,200]$ the mode of the system is 8, for $k\in[300,360]$ the mode of the system is 2, and for $k\in[480,580]$ the mode of the system is 7. Note that for comparison purposes, we also simulate a pure control formulation, i.e., \setcounter{equation}{32}
\begin{equation}\label{eq:OptProb3}
u_{0:N-1}^\ast=\left\{
\begin{array}{ll}
& \arg\,\min\limits_{u_{0:N-1}}\,J_c(\cdot)\text{ given in \eqref{eq:ControlOFFinal}}\\
\text{s.t.} & \eqref{eq:constraintsFinal}\text{ is satisfied  }\forall i,~\forall k\geq0
\end{array}
\right..
\end{equation}



The achieved normalized values of the control and detection objection functions are shown in Fig. \ref{fig:cost1} and \ref{fig:cost2}, where control and detection costs obtained by the formulation \eqref{eq:OptProb3} are assumed as the base unit quantity for control and detection costs, respectively. As expected, compromise between control and detection aims increases the control cost $J_c$. However, it decreases the detection cost $\hat{J}_d$ which means that probability of misidentification is minimized.




\section{Conclusion}\label{sec:conclusion}
This paper proposed an optimization approach for active attack detection and control of constrained CPSs systems subject to expectational linear constraints. This paper mainly focused on PA2 attack, where the attacker prevents the exchange of information between the controller and the actuators. A set of parallel detectors based on hypothesis testing approach was proposed. Using a probabilistic approach to deal with uncertainties, the detection and control aims were formulated as two separate stochastic objective functions. The open loop approach was deployed to transfer the stochastic functions to deterministic ones. Two alternative compromise between detection and control aims were presented in the form of a constrained optimization problem. The effectiveness of the proposed active approach was validated through simulation studies.

\begin{figure}[!t]
\centering
\includegraphics[width=8cm]{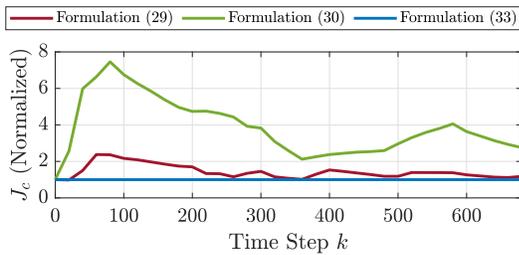}
\caption{Normalized control cost by formulations \eqref{eq:OptProb3}, \eqref{eq:OptProb1}, and \eqref{eq:OptProb2}.}
\label{fig:cost1}
\end{figure}

\begin{figure}[!t]
\centering
\includegraphics[width=8cm]{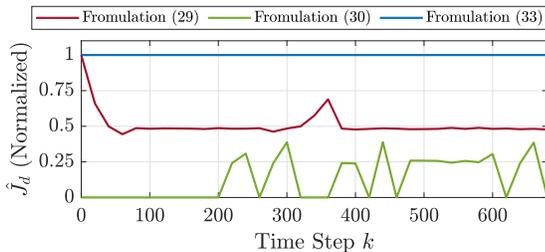}
\caption{Normalized detection cost by formulations \eqref{eq:OptProb3}, \eqref{eq:OptProb1}, and \eqref{eq:OptProb2}.}
\label{fig:cost2}
\end{figure}

\balance
\bibliographystyle{IEEEtran}
\bibliography{ref}{}

\end{document}